  \documentclass[a4paper]{article}

\usepackage{latexsym}
\usepackage{amsmath}
\usepackage{amssymb}
\usepackage[pdftex]{graphicx}

 \usepackage[english]{babel}
 \usepackage[naturalnames,colorlinks,pagebackref,final]{hyperref}
 \usepackage{amsthm}

\newcounter{problemcnt}
\newtheorem{problem}[problemcnt]{Problem}

\newtheorem{theorem}{Theorem}[section]
\newtheorem{corollary}[theorem]{Corollary}

\newtheorem{proposition}[theorem]{Proposition}
\newtheorem{Definition}[theorem]{Definition}
\newtheorem{example}[theorem]{Example}

\newtheorem{conjecture}[theorem]{Conjecture}

\renewcommand{\vec}[1]{\mathbf{#1}}
\newcommand{\tuple}[1]{\ensuremath{\langle #1 \rangle}}

\newcommand{\C}{\ensuremath{\mathcal{C}}}

\newcommand{\cone}[1]{\ensuremath{\operatorname{Cone}(#1)}}
\newcommand{\express}[1]{\ensuremath{\tuple{#1}}}
\newcommand{\cspo}{\textsc{CSP}}
\newcommand{\vcspo}{\textsc{VCSP}}
\newcommand{\sfm}{\textsc{SFM}}
\newcommand{\vcsp}[1]{\textsc{VCSP}(\ensuremath{#1})}

\newcommand{\RRo}{\overline{\mathbb{R}}}

\newcommand{\mul}[1]{\ensuremath{\mathsf{Mul}(#1)}}

\fussy

\begin{document}

\title{The Expressive Power of Binary Submodular Functions\thanks{An earlier
version of some parts of the results of this article appeared in {\em
Proceedings of the 14th International Conference on Principles and Practise of
Constraint Programming (CP)}, 2008, pp. 112--127, and in Oxford University
Computing Laboratory Technical Report CS-RR-08-08, June 2008.}}

\author{Stanislav \v{Z}ivn\'y \\
Oxford University \\
Computing Laboratory \\
\texttt{stanislav.zivny@comlab.ox.ac.uk}
\and
David A. Cohen\\
Royal Holloway, University of London\\
Department of Computer Science\\
\texttt{dave@cs.rhul.ac.uk}\\
\and
Peter G. Jeavons \\
Oxford University \\
Computing Laboratory \\
\texttt{peter.jeavons@comlab.ox.ac.uk}}
\date{}
\maketitle

\begin{abstract}

It has previously been an open problem whether all Boolean submodular functions
can be decomposed into a sum of binary submodular functions over a possibly
larger set of variables. This problem has been considered within several
different contexts in computer science, including computer vision, artificial
intelligence, and pseudo-Boolean optimisation. Using a connection between the
expressive power of valued constraints and certain algebraic properties of
functions, we answer this question negatively.

Our results have several corollaries. First, we characterise precisely which
submodular functions of arity 4 can be expressed by binary submodular functions.
Next, we identify a novel class of submodular functions of arbitrary arities
which can be expressed by binary submodular functions, and therefore minimised
efficiently using a so-called expressibility reduction to the {\sc Min-Cut}
problem. More importantly, our results imply limitations on this kind of
reduction and establish for the first time that it cannot be used in general to
minimise arbitrary submodular functions. Finally, we refute a conjecture of
Promislow and Young on the structure of the extreme rays of the cone of Boolean
submodular functions.

\end{abstract}

\noindent{\bf Keywords:}
Combinatorial optimisation, decomposition of submodular functions, expressive
power, Gibbs energy minimisation, Markov Random Fields, min cut, multimorphisms,
submodular pseudo-Boolean minimisation, submodular polynomials, valued
constraint satisfaction problems.

\section{Introduction}

\subsection{Background}
\label{sec:background}

A function $f:2^V\rightarrow\mathbb{R}$ is called {\em submodular\/} if for all $S,T\subseteq V$,
\[
f(S\cap T)+f(S\cup T)\ \leq\ f(S)+f(T).
\]
Submodular functions are a key concept in operational research and combinatorial
optimisation~\cite{Nemhauser88:optimization,Narayanan97:submodular,Topkis98:book,Schrijver03:CombOpt,Fujishige05:submodular2nd,Korte07:combopt,Iwata08:sfm-survey}.
Examples include cut capacity functions, matroid rank functions, and entropy
functions. Submodular functions are often considered as a discrete analogue of
convex functions~\cite{Lovasz83:submodular}.

Both minimising and maximising submodular functions, possibly under some
additional conditions, have been considered extensively in the literature.
Submodular function maximisation is easily shown to be
NP-hard~\cite{Schrijver03:CombOpt} since it generalises many standard NP-hard
problems such as the maximum cut problem. In contrast, the problem of {\em
minimising\/} a submodular function ($\sfm$) can be solved efficiently with only
polynomially many oracle calls, either by using the ellipsoid
algorithm~\cite{Grotschel81:ellipsoid,Grotschel88:geometric}, or by using one of
several combinatorial algorithms that have been obtained in the last
decade~\cite{Schrijver00:submodular,Iwata01:Combinatorial,Iwata02:submodular,Iwata03:scaling,Orlin07:faster,Iwata09:simple}.
The time complexity of the fastest known general algorithm for SFM is
$O(n^6+n^5L)$, where $n$ is the number of variables and $L$ is the time required
to evaluate the function~\cite{Orlin07:faster}.

The minimisation of submodular functions on sets is equivalent to the
minimisation of submodular functions on distributive lattices~\cite{Schrijver03:CombOpt}.
Krokhin and Larose have also studied the more general problem of minimising submodular functions on
non-distributive lattices~\cite{KrokhinL08:max}.

An important and well-studied sub-problem of $\sfm$ is the minimisation of
submodular functions of bounded arity ($\mbox{\sc SFM}_b$), also known as {\em
locally defined\/} submodular functions. In this scenario the submodular
function to be minimised is defined as the sum of a collection of functions
which each depend only on a bounded number of variables. Locally defined
optimisation problems occur in a variety of contexts: \begin{itemize} \item In
the context of pseudo-Boolean optimisation, such problems involve the
minimisation of Boolean polynomials of bounded {\em
degree}~\cite{Boros:pseudo-boolean}. \item In the context of artificial
intelligence, they have been studied as {\em valued constraint satisfaction
problems} ($\vcspo$)~\cite{Rossi06:handbook}, also known as {\em soft\/} or {\em
weighted\/} constraint satisfaction problems. \item In the context of computer
vision, such problems are often formulated as {\em Gibbs energy minimisation\/}
problems~\cite{Geman84} or {\em Markov Random
Fields}~\cite{Lauritzen96,Wainwright}. \end{itemize}
We will present our results primarily in the language of pseudo-Boolean
optimisation. Hence an instance of $\mbox{\sc SFM}_b$ with $n$ variables will be
represented as a polynomial in $n$ Boolean variables, of some fixed bounded
degree. However, we will also mention the consequences of our results for
constraint satisfaction problems and certain optimisation problems arising in
computer vision.

A general algorithm for $\sfm$ can always be used for the more restricted
$\mbox{\sc SFM}_b$, but the special features of this more restricted problem
sometimes allow more efficient special-purpose algorithms to be used. (Note that
we are focusing on {\em exact\/} algorithms which find an optimum solution.) In
particular, it has been shown that certain cases can be solved much more
efficiently by reducing to the {\sc Min-Cut} problem, that is, the problem of
finding a minimum cut in a directed graph which includes a given source vertex
and excludes a given target vertex. For example, it has been known since 1965
that the minimisation of {\em quadratic\/} submodular polynomials is equivalent to
finding a minimum cut in a corresponding directed
graph~\cite{Hammer65:bin,Boros:pseudo-boolean}. Hence quadratic submodular
polynomials can be minimised in $O(n^3)$ time, where $n$ is the number of
variables.

A similar approach, using a reduction to {\sc Min-Cut},
can be used for any class of polynomials which can be decomposed into a
{\em sum\/} of quadratic submodular polynomials, perhaps with additional variables to be minimised over.
We will say that a polynomial that can be decomposed in this way is {\em
expressible\/} by
quadratic submodular polynomials (see Section~\ref{sec:background}).
The following classes of functions have all been shown to be expressible in this way,
over the past four decades:
\begin{itemize}
\item
polynomials where all terms of degree 2 or more have negative coefficients (also known as negative-positive polynomials)~\cite{Rhys70:selection};
\item
cubic submodular polynomials~\cite{Bill85:Maximizing};
\item
$\{0,1\}$-valued submodular functions (also known as 2-monotone functions)~\cite{Creignouetal:siam01,Cohen05:supermodular};
\item
binary submodular functions over non-Boolean domains~\cite{Cohen04:soft}
(also known as Monge matrices~\cite{Burkard96:Monge});
\item
generalised 2-monotone functions~\cite{Cohen05:supermodular};
\item
a class recently found by~\v{Z}ivn\'{y} and Jeavons~\cite{zj08:cp} and independently by Zalesky~\cite{Zalesky08:submodular}.
\end{itemize}
All these classes of functions have
been shown to be expressible by quadratic submodular polynomials
and hence can be minimised in cubic time.

This series of positive expressibility results naturally raises the following question:
%
\begin{problem} \label{probl}
Are {\em all\/} submodular polynomials expressible by quadratic submodular polynomials?
\end{problem}

Each of the above expressibility results was obtained by an ad-hoc construction,
and no general technique has previously been
proposed which is sufficiently powerful to address Problem~\ref{probl}.

\subsection{Contributions}

Cohen et al. recently
developed a novel algebraic approach to characterising the expressive power of
valued constraints in terms of certain algebraic properties of those constraints~\cite{Cohen06:expressive}.

Using this systematic algebraic approach we are able to give a negative answer to
Problem~\ref{probl}: we show that there are quartic submodular polynomials which
are not expressible by quadratic submodular polynomials. More precisely, we
characterise exactly which quartic submodular polynomials are expressible by
quadratic submodular polynomials and which are not.
In addition, we show that any quartic submodular polynomial
is either expressible by quadratic submodular polynomials with only linearly
many extra variables, or it is not expressible at all.

On the way to establishing this result we show that two broad families of submodular functions
known as {\em upper fans\/} and {\em lower fans\/} are all
expressible by binary submodular functions.
This provides a new class of submodular polynomials of all arities
which are expressible by quadratic submodular polynomials and hence solvable efficiently
by reduction to {\sc Min-Cut}. We use the expressibility of this family,
and the existence of non-expressible functions, to refute a conjecture
from~\cite{Promislow05:supermodular} on the structure of the extreme rays of the
cone of Boolean submodular functions, and suggest a more refined conjecture of our own.

\subsection{Applications}

The concept of submodularity is important in a wide variety of fields
within computer science; in this paper we briefly discuss two of these:
artificial intelligence and computer vision.
Our results can be directly applied to both of these areas,
as we show in Section~\ref{sec:applications} below.

\paragraph{Artificial Intelligence}

A major area of investigation in artificial intelligence is
the {\em Constraint Satisfaction\/} problem ($\cspo$)~\cite{Rossi06:handbook}.
A number of extensions have been added to the basic $\cspo$ framework
to deal with questions of optimisation, including
semi-ring $\cspo$s, valued $\cspo$s, soft $\cspo$s and weighted $\cspo$s.
These extended frameworks can be used to model a wide range of
discrete optimisation
problems~\cite{Schiex95:valued,Bistarelli99:semiring+VCSP,Rossi06:handbook},
including standard problems such as {\sc Min-Cut}, {\sc Max-Sat},
{\sc Max-Ones Sat}, {\sc Max-CSP}~\cite{Creignouetal:siam01,Cohen06:complexitysoft},
and {\sc Min-Cost Homomorphism}~\cite{Gutin06:repair}.

The differences between the various frameworks are not relevant for our purposes,
so we will simply focus on one very general framework, the valued constraint satisfaction problem or $\vcspo$.
Informally, in the $\vcspo$ framework, an instance consists of a set of
variables, a set of possible values for those variables, and a set of constraints. Each constraint
has an associated cost function which assigns a cost (or degree of
violation) to every possible tuple of values for the variables in the scope of
the constraint. The goal is to find an assignment of values to all of the
variables which has the minimum total cost.

The class of constraints with submodular cost functions is the only non-trivial
tractable class of optimisation problems in the dichotomy classification of the
Boolean $\vcspo$~\cite{Cohen06:complexitysoft}, and the only tractable class in
the dichotomy classification of the {\sc Max-CSP} problem for both 3-element
sets~\cite{JonssonKK06:3valued} and arbitrary finite sets allowing constant
(that is, fixed-value) constraints~\cite{DeinekoJKK08}.

Cohen et al. showed that $\vcspo$ instances with submodular constraints over an
arbitrary finite domain can be reduced to $\sfm$~\cite{Cohen06:complexitysoft},
and hence can be solved in polynomial time. This tractability result has since
been generalised to a wider class of valued constraints over arbitrary finite
domains known as tournament-pair constraints~\cite{Cohen08:Generalising}. An
alternative approach to solving $\vcspo$ instances with bounded-arity submodular
constraints, based on linear programming, can be found
in~\cite{Cooper08:minimizing}.

\paragraph{Computer Vision}

Gibbs energy minimisation and Markov Random Fields,
play an important role in computer vision as they are applicable to a
wide variety of vision problems, including image restoration, stereo vision and
motion tracking, image synthesis, image segmentation, multi-camera scene
reconstruction and medical imaging~\cite{KolmogorovZ04}.
Reducing energy minimisation to the {\sc Min-Cut} problem has recently become
a very popular approach, leading to the rediscovery of
the property of submodularity~\cite{KolmogorovZ04,Freedman05:sub},
and showing that certain special classes of functions can be
minimised using graph cuts by introducing extra variables~\cite{Ramalingam08:exact,Hokli08:graph}.

Our results below characterise precisely which 4-ary submodular functions can be
minimised using graph cuts in this way and which cannot.
We also provide a new class of submodular functions of
arbitrary arity which can be minimised efficiently in this way.

\section{Preliminaries}
\label{sec:premlims}
In this section, we introduce the basic definitions and the main tools used
throughout the paper.
\subsection{Cost functions and expressibility}
We denote by $\RRo$ the set of all real numbers together with (positive) infinity.
For any fixed set $D$, a function $\phi$ from $D^n$ to $\RRo$ will be called a
{\em cost function\/} on $D$ of arity $n$.
If the range of $\phi$ lies entirely within $\mathbb{R}$,
then $\phi$ is called a {\em finite-valued\/} cost function.
If the range of $\phi$ is $\{0,\infty\}$, then $\phi$ can be viewed as a predicate,
or {\em relation}, allowing just those tuples $t \in D^n$ for which $\phi(t) = 0$.

Cost functions can be added and multiplied by arbitrary real values,
hence for any given set of cost functions, $\Gamma$, we can define
the convex cone generated by $\Gamma$, as follows.
\begin{Definition}
For any set of cost functions $\Gamma$, the {\em cone generated by $\Gamma$},
denoted $\cone{\Gamma}$, is defined by:
\[
\cone{\Gamma} =
\{\alpha_1 \phi_1 + \cdots \alpha_r \phi_r \mid r \ge 1;\ \phi_1,\ldots,\phi_r \in \Gamma;\ \alpha_1,\ldots\alpha_r \ge 0\}.
\]
\end{Definition}

\begin{Definition}
A cost function $\phi$ of arity $n$ is said to be {\em expressible\/} by a set of cost functions $\Gamma$
if $\phi=\min_{y_1,\ldots,y_j} \phi'(x_1,\ldots,x_n,y_1,\ldots,y_j) + \kappa$,
for some $\phi' \in \cone{\Gamma}$ and some constant $\kappa$.

The variables $y_1,\ldots,y_j$ are called {\em extra\/} (or {\em hidden}) variables, and
$\phi'$ is called a {\em gadget\/} for $\phi$ over $\Gamma$.
\end{Definition}

Note that in the special case of relations  this notion of expressibility
corresponds to the standard notion of expressibility using conjunction and
existential quantification ({\em primitive positive formulas})~\cite{Bulatov05:classifying}.

We denote by $\express{\Gamma}$ the {\em expressive power\/} of $\Gamma$, which
is the set of all cost functions expressible by $\Gamma$.

It was shown in~\cite{Cohen06:expressive} that the expressive power of a set of
cost functions is determined by certain algebraic properties of those cost
functions called fractional polymorphisms. For the results of this paper, we
will only need a certain subset of these algebraic properties, called {\em
multimorphisms}~\cite{Cohen06:complexitysoft}. These are defined in
Definition~\ref{def:multim} below, which is illustrated in
Figure~\ref{fig:multim}.

The $i$-th component of a tuple $t$ will be denoted by $t[i]$. Note that any
operation on a set $D$ can be extended to tuples over the set $D$ in a standard
way, as follows. For any function $f: D^k \rightarrow D$, and any collection of
tuples ${t_1,\ldots,t_k \in D^n}$, define ${f(t_1,\ldots,t_k) \in D^n}$ to be
the tuple ${\tuple{f(t_1[1],\ldots,t_k[1]),\ldots,f(t_1[n],\ldots,t_k[n]) }}.$

\begin{Definition}[\cite{Cohen06:complexitysoft}] \label{def:multim}

Let $\mathcal{F}:D^k\rightarrow D^k$ be the function whose $k$-tuple of output values is given
by the tuple of functions $\mathcal{F}=\tuple{f_1,\ldots,f_k}$, where each $f_i:D^k\rightarrow D$.

For any $n$-ary cost function $\phi$, we say that $\mathcal{F}$ is a $k$-ary
{\em multimorphism\/} of $\phi$ if, for all $t_1,\ldots,t_k \in D^n$,
\[ \sum_{i=1}^k \phi(t_i) \ \geq \ \sum_{i=1}^k \phi(f_i(t_1,\ldots,t_k)).
\]
\end{Definition}

\begin{figure}[tbp]
  \[
 \begin{array}{c}
  \begin{array}{c}
  t_1\\
  t_2\\
  \vdots\\
  t_k\\
  \end{array}
  \\
  \begin{array}{c}
    \
  \end{array}
  \\
  \begin{array}{c}
  t_1'=f_1(t_1,\ldots,t_k)\\
  t_2'=f_2(t_1,\ldots,t_k)\\
  \vdots\\
  t_k'=f_k(t_1,\ldots,t_k)\\
  \end{array}
  \end{array}
  \begin{array}{c}
  \begin{array}{cccccc}
    t_1[1] & t_1[2] & \ldots & t_1[n] \\
    t_2[1] & t_2[2] & \ldots & t_2[n] \\
    & & \vdots & \\
    t_k[1] & t_k[2] & \ldots & t_k[n] \\
  \end{array}
  \\
  \hline
  \begin{array}{c}
    \
  \end{array}
  \\
  \begin{array}{cccccc}
    t'_1[1] & t'_1[2] & \ldots & t'_1[n] \\
    t'_2[1] & t'_2[2] & \ldots & t'_2[n] \\
    & & \vdots & \\
    t'_k[1] & t'_k[2] & \ldots & t'_k[n] \\
  \end{array}
  \\
  \end{array}
  \begin{array}{c}
  \stackrel{\phi}{\longrightarrow}
  \left.
  \begin{array}{c}
    \phi(t_1)\\
    \phi(t_2)\\
    \vdots\\
    \phi(t_k)\\
  \end{array}
  \right\}\mbox{\normalsize{$\displaystyle \sum_{i=1}^{k}\phi(t_i)$}}
  \\
  \begin{array}{c}
    \qquad\qquad\qquad \mbox{\rotatebox{270}{$\geq$}}
  \end{array}
  \\
  \ \stackrel{\phi}{\longrightarrow}
  \left.
  \begin{array}{c}
    \phi(t'_1)\\
    \phi(t'_2)\\
    \vdots\\
    \phi(t'_k)\\
  \end{array}
  \right\}\mbox{\normalsize{$\displaystyle \sum_{i=1}^{k}\phi(t'_i)$}}
  \\
  \end{array}
  \]
  \caption{Inequality establishing $\mathcal{F} = \tuple{f_1,\ldots,f_k}$
  as a multimorphism of cost function $\phi$ (see Definition~\ref{def:multim}).}
  \label{fig:multim}
\end{figure}

For any set of cost functions, $\Gamma$, we will say that $\mathcal{F}$ is a
multimorphism of $\Gamma$ if $\mathcal{F}$ is a multimorphism of every cost
function in $\Gamma$. The set of all multimorphisms of $\Gamma$ will be denoted
$\mul{\Gamma}$.

Note that multimorphisms are preserved under expressibility. In other words, if
$\mathcal{F}\in\mul{\Gamma}$, and $\phi\in\express{\Gamma}$, then
$\mathcal{F}\in\mul{\{\phi\}}$~\cite{Cohen06:complexitysoft,Cohen06:expressive}.
This has two important corollaries. First, if
$\express{\Gamma_1}=\express{\Gamma_2}$, then $\mul{\Gamma_1}=\mul{\Gamma_2}$.
Second, if there exists $\mathcal{F}\in\mul{\Gamma}$ such that
$\mathcal{F}\not\in\mul{\{\phi\}}$, then $\phi$ is not expressible
over $\Gamma$, that is, $\phi\not\in\express{\Gamma}$.

\subsection{Lattices and submodularity}

Recall that $L$ is a {\em lattice\/} if $L$ is a partially ordered set in which every
pair of elements $(a,b)$ has a unique supremum (the least upper bound of $a$ and $b$,
called the {\em join}, denoted $ a\vee b$) and a unique infimum (the greatest lower bound,
called the {\em meet}, denoted $a\wedge b$).

For any lattice-ordered set $D$, a cost function $\phi:D^n\rightarrow\RRo$ is
called {\em submodular\/} if for every $u,v\in D^m$,
$\phi(\min(u,v))+\phi(\max(u,v))\le\phi(u)+\phi(v)$ where both $\min$ and $\max$
are applied coordinate-wise on tuples $u$ and
$v$~\cite{Nemhauser88:optimization}.
This standard definition can be reformulated very simply in terms of multimorphisms:
$\phi$ is submodular if $\tuple{\min,\max}\in\mul{\{\phi\}}$.

Using results from~\cite{Schrijver03:CombOpt} and~\cite{Cohen06:complexitysoft},
it can be shown that any submodular cost function $\phi$ can be expressed as the
sum of a finite-valued submodular cost function $\phi_{fin}$, and a submodular
relation $\phi_{rel}$, that is, $\phi=\phi_{fin}+\phi_{rel}$.

Moreover, it is known that all submodular {\em relations\/} are binary
decomposable~\cite{Jeavons98:consist}, and hence expressible using only binary
submodular relations. Therefore,  when considering which cost functions are
expressible by binary submodular cost functions, we can restrict our attention
to {\em finite-valued\/} cost functions without any loss of generality.

Next we define some particular families of submodular cost functions,
first described in~\cite{Promislow05:supermodular},
which will turn out to play a central role in our analysis.

\begin{Definition}\label{def:upperfan}
Let $L$ be a lattice. We define the following cost functions on $L$:
\begin{itemize}
\item
For any set $F$ of pairwise incomparable elements
$(a_1,\ldots,a_m) \subseteq L$, such that each pair of distinct elements $(a_i,a_j)$ has the
same least upper bound, $\bigvee F$,
the following cost function is called an {\em upper fan}:
\[
\phi_F(x) =
\begin{cases}
-2           & \mbox{if\ } x\geq \bigvee F,\\
-1           & \mbox{if $x\not\geq \bigvee F$, but $x \geq a_i$ for some $i$}, \\
\phantom{-}0 & \mbox{otherwise}.
\end{cases}
\]
\item
For any set $G$ of pairwise incomparable elements
$(a_1,\ldots,a_m) \subseteq L$, such that each pair of distinct elements $(a_i,a_j)$ has the
same greatest lower bound, $\bigwedge G$,
the following cost function is called a {\em lower fan}:
\[
\phi_G(x) =
\begin{cases}
-2           & \mbox{if\ } x\leq \bigwedge G,\\
-1           & \mbox{if $x\not\leq \bigwedge G$, but $x \leq a_i$ for some $i$}, \\
\phantom{-}0 & \mbox{otherwise}.
\end{cases}
\]
\end{itemize}
\end{Definition}

We call a cost function a {\em fan\/} if it is either an upper fan or a lower
fan. It is not hard to show that all fans are
submodular~\cite{Promislow05:supermodular}.

Note that our definition of fans is slightly more general than the
definition in~\cite{Promislow05:supermodular}. In particular, we allow the set $F$ to be empty,
in which case the corresponding upper fan $\phi_F$ is a constant function.

\subsection{Boolean cost functions and polynomials}

In this paper we will focus on problems over Boolean domains, that is, where $D=\{0,1\}$.

Any cost function of arity $n$ can be represented as a table of values of size $D^n$.
Moreover, a finite-valued cost function $\phi:D^n\rightarrow\mathbb{R}$ on a Boolean domain $D = \{0,1\}$
can also be represented as a unique {\em polynomial\/} in $n$ (Boolean) variables
with coefficients from $\mathbb{R}$ (such functions are sometimes called {\em
pseudo-Boolean functions}~\cite{Boros:pseudo-boolean}). Hence, in what follows,
we will often refer to a finite-valued cost function on a Boolean domain and its
corresponding polynomial interchangeably.

For polynomials over Boolean variables there is a standard way to define
{\em derivatives\/} of each order (see~\cite{Boros:pseudo-boolean}).
For example, the second order derivative of a polynomial $p$, with respect to the first two
indices,  denoted $\delta_{1,2}(\vec{x})$, is defined as
$p(1,1,\vec{x})-p(1,0,\vec{x})-p(0,1,\vec{x})+p(0,0,\vec{x})$.
Analogously for all other pairs of indices.
It was shown in~\cite{Fisher:sub} that a polynomial $p(x_1,\ldots,x_n)$ over Boolean
variables $x_1, \ldots, x_n$ represents a submodular cost function if and only
if its second order derivatives $\delta_{i,j}(\vec{x})$ are non-positive for all
$1\le i<j\le n$ and all $\vec{x}\in D^{n-2}$.
An immediate corollary is that a quadratic polynomial represents a submodular
cost function if and only if the coefficients of all quadratic terms are non-positive.

Note that a cost function is called {\em supermodular\/} if all its second order
derivatives are non-negative. Clearly, $f$ is submodular if and only if $-f$ is
supermodular. Cost functions which are both submodular and supermodular (in
other words, all second order derivatives are equal to zero) are called {\em
modular}, and polynomials corresponding to modular cost functions are
linear~\cite{Boros:pseudo-boolean}.

\begin{example}
\label{ex:upperpolynomial}
For any set of indices $I=\{{i_1},\ldots,{i_m}\} \subseteq \{1,\ldots,n\}$
we can define a cost function $\phi_I$ in $n$ variables as follows:
\[
\phi_I(x_1,\ldots,x_n)\ =\
\begin{cases}
-1           & \mbox{if\ } (\forall i\in I)(x_i=1), \\
\phantom{-}0 & \mbox{otherwise}.
\end{cases}
\]
The polynomial representation of $\phi_I$ is $p(x_1,\ldots,x_n)=-x_{i_1}\ldots x_{i_m}$,
which is a polynomial of degree $m$.
Note that it is straightforward to verify that $\phi_I$ is submodular
by checking the second order derivatives of $p$.

However, the function $\phi_I$ is also expressible by {\em quadratic\/} polynomials,
using a single extra variable, $y$, as follows:
\[
\phi_I(x_1,\ldots,x_n)=\min_{y\in\{0,1\}}\{-y + y \sum_{i\in I}(1-x_i)\}.
\]
We remark that this is a special case of the expressibility result for negative-positive polynomials
first obtained in~\cite{Rhys70:selection}.
\end{example}

Note that when $D= \{0,1\}$, the set $D^n$ with the product ordering is isomorphic to the
lattice of all subsets of an $n$-element set ordered by inclusion.
Hence, a cost function on a Boolean domain can be viewed as a cost function
defined on a lattice of subsets, and
we can apply Definition~\ref{def:upperfan} to identify certain Boolean functions
as upper fans or lower fans, as the following example indicates.
\begin{example}
\label{ex:upperfanpoly}
Let $F=\{I_1,\ldots,I_r\}$ be a set of subsets of $\{1,2,\dots,n\}$ such that
for all $i \neq j$ we have $I_i \not\subseteq I_j$ and $I_i \cup I_j = \bigcup F$.

By Definition~\ref{def:upperfan}, the corresponding upper fan function $\phi_F$
has the following polynomial representation:
\[
p(x_1,\ldots,x_n) = (r-2)\prod_{i\in \bigcup F} x_i - \prod_{i\in I_1} x_i - \cdots - \prod_{i\in I_r} x_i.
\]
\end{example}

We remark that any permutation of a set $D$ gives rise to an automorphism of
cost functions over $D$.
In particular, for any cost function $f$ on a Boolean domain $D$,
the {\em dual\/} of $f$ is the corresponding cost function which results from exchanging
the values 0 and 1 for all variables. In
other words, if $p$ is the polynomial representation of $f$, then the dual of
$f$ is the cost function whose polynomial representation is obtained from $p$
by replacing all variables $x$ with $1-x$. Observe that, due to symmetry,
taking the dual preserves submodularity and expressibility by binary submodular cost functions.

It is not hard to see that upper fans are duals of lower fans and vice versa.

\section{Results} \label{sec:result}

In this section, we present our main results. First, we show that fans of all
arities are expressible by binary submodular cost functions. Next, we
characterise the multimorphisms of binary submodular cost functions. Finally,
combining these results together, we characterise precisely which 4-ary
submodular cost functions are expressible by binary submodular cost functions.
More importantly, we show that some submodular cost functions are {\em not\/}
expressible by binary submodular cost functions, and therefore cannot be
minimised using the {\sc Min-Cut} problem via an expressibility reduction.
Finally, we describe some applications of these results to valued constraint
satisfaction problems and certain optimisation problems arising in computer vision.

\subsection{Expressibility of upper fans and lower fans}

We denote by $\Gamma_{\sf sub,n}$ the set of all finite-valued
submodular cost functions of arity at most $n$ on a Boolean domain $D$,
and we set $\Gamma_{\sf sub}=\bigcup_n\Gamma_{\sf sub,n}$.

We denote by $\Gamma_{\sf fans,n}$ the set of all fans of arity at most $n$ on
a Boolean domain $D$, and we
set $\Gamma_{\sf fans}=\bigcup_n\Gamma_{\sf fans,n}$.

Our next result shows that $\Gamma_{\sf fans}\subseteq\express{\Gamma_{\sf sub,2}}$.
\begin{theorem} \label{thm:upper}
Any fan on a Boolean domain $D$ is expressible by binary submodular functions on $D$
using at most $1 + \lfloor m/2\rfloor$ extra variables,
where $m$ is the degree of its polynomial representation.
\end{theorem}
\begin{proof}
Since upper fans are dual to lower fans, it is sufficient to establish the
result for upper fans only.

Let $F=\{I_1,\ldots,I_r\}$ be a set of subsets of $\{1,2,\dots,n\}$ such that
for all $i \neq j$ we have $I_i \not\subseteq I_j$ and $I_i \cup I_j = \bigcup F$,
and let $\phi_F$ be the corresponding upper fan, as specified by
Definition~\ref{def:upperfan}. The polynomial representation of $\phi_F$, $p(x_1,\ldots,x_n)$,
is given in Example~\ref{ex:upperfanpoly}.

The degree of $p$ is equal to the total number of variables occurring in it, which will be denoted $m$.
Note that $m = |\bigcup F|$.

If $r=0$, then $\phi_F$ is constant, so the result holds trivially.
If $r=1$, we have $F=\{I\}$, where $I=\{i_1,\ldots,i_m\}$ and the
polynomial representation of $\phi_F$ is $-2x_{i_1}x_{i_2}\cdots x_{i_m}$.
In this case, it was shown in Example~\ref{ex:upperpolynomial} that $\phi_F$
can be expressed by quadratic functions using one extra variable, as follows:
\[
-2x_{i_1}x_{i_2}\cdots x_{i_m} = \min_{y\in\{0,1\}} \{ 2y((m-1) - \sum_{i \in I} x_{i} )\}.
\]

%
For the case when $r > 1$,
we first note that any $i\in \bigcup F$ must belong to all the elements of $F$ except for at most one
(otherwise there would be two elements of $F$, say $I_i$ and $I_j$, such that $I_i \cup I_j \neq \bigcup F$,
which contradicts the choice of $F$).

We will say that two elements of $\bigcup F$ are {\em equivalent\/} if they occur in exactly the same elements
of $F$, that is, $i_1, i_2 \in \bigcup F$ are equivalent if
$i_1 \in I_j \Leftrightarrow i_2\in I_j$ for all $j \in \{i,\ldots,r\}$.
Equivalent elements $i_1$ and $i_2$ of $\bigcup F$ can be merged by replacing them with a single new element.
In the polynomial representation of $\phi_F$ this corresponds to replacing
the variables $x_{i_1}$ and $x_{i_2}$ with a single new variable, $z$, corresponding to their product.
Note that the number of equivalence classes of size two or greater is at most $\lfloor m/2\rfloor$.

After completing all such merging, we obtain a new set
$F'=\{I'_1,\ldots,I'_{r'}\}$ with the property that $|I'_i|=m'-1$ for every $i$,
where $m'=|\bigcup F'|$ is the size of the common join of any $I'_i,I'_j\in F'$.
This set has a corresponding new upper fan, $\phi_{F'}$, over the new merged variables.

To complete the proof we will construct a simple gadget for expressing $\phi_{F'}$, and show how to
use this to obtain a gadget for expressing the original upper fan $\phi_F$.

Note that the sets $I'_i$ are subsets of $\bigcup F'$, each of size $m'-1$.
Any such subset is uniquely determined by its single missing element.
We denote by $K$ the set of elements occurring in {\em all\/} sets $I'_i$
and by $L$ the set of elements which are missing from one of these subsets.
Clearly, $|K|+|L|=m'$.
We claim that the following polynomial is a gadget for expressing $\phi_F'$:
\[
p'(z_1,\ldots,z_{m'}) = \min_{y\in\{0,1\}} \{ y (2(m'-1) - |L| - \sum_{i\in L} z_i - 2 \sum_{i\in K} z_i)\}.
\]
To establish this claim, we will compute the value of $p'$,
for each possible assignment to the variables $z_1,\ldots,z_{m'}$.
Denote by $k_0$ the number of $0$s assigned to variables in $K$,
and by $l_0$ the number of $0$s assigned to variables in $L$.
Then we have:
\begin{equation*}
\begin{split}
p'(z_1,\ldots,z_{m'}) =\
& \min_{y\in\{0,1\}} y(2m'-2-|L|-\sum_{i\in L}z_i-2\sum_{i\in K}z_i) \\
=\ & \min_{y\in\{0,1\}} y(2m'-2-|L|-(|L|-l_0)-2(m'-|L|-k_0) \\
=\ & \min_{y\in\{0,1\}} y(2m'-2-2|L|+l_0-2m'+2|L|+2k_0) \\
=\ & \min_{y\in\{0,1\}} y(-2+2k_0+l_0).\\
\end{split}
\end{equation*}
Hence if $k_0=l_0=0$, then $p'$ takes the value -2. If $k_0=0$ and $l_0=1$, then
$p'$ takes the value -1. In all other cases (that is, $k_0>0$ or $l_0>1$), $p'$
takes the value 0. By Definition~\ref{def:upperfan}, this means that $p'$ is the
(unique) polynomial representation for $\phi_{F'}$. Note that $p'$ uses just one
extra variable, $y$.

Finally, we show how to obtain a gadget for the original upper fan $\phi_F$,
from the polynomial $p'$. Each variable in $p'$ represents an equivalence class of
elements of $\bigcup F$, so it can be replaced by a
term consisting of the product of the variables in this equivalence class.
In this way we obtain a new polynomial over the original variables containing
linear and negative quadratic terms
together with negative higher order terms (cubic or above) corresponding to every
equivalence class with 2 or more elements.
However, each of these higher order terms can itself be expressed by a quadratic
submodular polynomial,
by introducing a single extra variable, as shown in the case when $r=1$, above.
Therefore, combining each of these polynomials,
the total number of new variables introduced is at most
$1+\lfloor m/2\rfloor$.
\end{proof}
Many of the earlier expressibility results mentioned in Section~\ref{sec:background}
can be obtained as simple corollaries of Theorem~\ref{thm:upper}, as the following examples indicate.
\begin{example}
Any negative monomial $-x_1x_2\cdots x_m$ is a positive multiple of an upper fan,
and the positive linear monomial $x_1$ is equal to $-(1-x_1) + 1$, so it is
a positive multiple of a lower fan, plus a constant.
Hence, by
Theorem~\ref{thm:upper}, all negative-positive submodular polynomials are expressible by
quadratic submodular polynomials, as originally shown in~\cite{Rhys70:selection}.
\end{example}
\begin{example}
Any cubic submodular polynomial can be expressed as
a positive sum of upper fans~\cite{Promislow05:supermodular}. Hence, by
Theorem~\ref{thm:upper}, all cubic submodular polynomials are expressible by
quadratic submodular polynomials, as originally shown in~\cite{Bill85:Maximizing}.
\end{example}
\begin{example}
A Boolean cost function $\phi$ is called {\em 2-monotone}~\cite{Creignouetal:siam01} if there exist two sets
$A,B\subseteq\{1,\ldots,n\}$ such that $\phi(\vec{x})=0$ if $A\subseteq\vec{x}$
or $\vec{x}\subseteq B$ and $\phi(\vec{x})=1$ otherwise (where
$A\subseteq\vec{x}$ means $\forall i\in A, x[i]=1$ and $\vec{x}\subseteq B$ means
$\forall i\not\in B, x[i]=0$).
It was shown in~\cite[Proposition 2.9]{Cohen05:supermodular} that a 2-valued Boolean cost function
is 2-monotone if and only if it is submodular.\footnote{In fact,
\cite{Cohen05:supermodular} studied {\em supermodular\/} cost functions, but
as $f$ is supermodular if and only if $-f$ is submodular, the results translate
easily.}

For any 2-monotone cost function defined by the sets of indices $A$ and $B$,
it is straightforward to check that
$\phi=\min_{y\in\{0,1\}}y(1+\phi_F/2)+(1-y)(1+\phi_G/2)$ where
$\phi_F$ is the upper fan defined by $F=\{A\}$ and
$\phi_G$ is the lower fan defined by $G=\{\overline{B}\}$.
Note that the function $y \phi_F$ is an upper fan, and the function $(1-y) \phi_G$ is a lower fan.
Hence, by
Theorem~\ref{thm:upper}, all 2-monotone polynomials are expressible by
quadratic submodular polynomials, and solvable by reduction to {\sc Min-Cut},
as originally shown in~\cite{Creignouetal:siam01}.
\end{example}
However, Theorem~\ref{thm:upper} also provides many new functions of all arities
which have not previously been
shown to be expressible by quadratic submodular functions, as the following example indicates.
\begin{example}
The function
$2x_1x_2x_3x_4-x_1x_2x_3-x_1x_2x_4-x_1x_3x_4-x_2x_3x_4$
belongs to $\Gamma_{\sf fans,4}$, but does not belong to any class of submodular functions
which has previously been shown to be expressible by quadratic submodular functions.
In particular, it does not belong to the class $\Gamma_{\sf new}$
identified in~\cite{zj08:cp,Zalesky08:submodular}.
\end{example}

\subsection{Characterising $\mul{\Gamma_{\sf sub,2}}$}

Since we have seen that a cost function can only be expressed by a given set of cost functions if it
has the same multimorphisms, we now investigate the multimorphisms of $\Gamma_{\sf sub,2}$.

A function $\mathcal{F}:D^k\rightarrow D^k$ is called {\em conservative\/} if,
for each possible choice of $x_1,\ldots,x_k$, the tuple
$\mathcal{F}(x_1,\ldots,x_k)$ contains the same multi-set of values,
$x_1,\ldots,x_k$ (in some order).

For any two tuples $\mathbf{x}=\tuple{x_1,\ldots,x_k}$ and
$\mathbf{y}=\tuple{y_1,\ldots,y_k}$ over $D$, we denote by
$H(\mathbf{x},\mathbf{y})$ the {\em Hamming distance\/} between $\mathbf{x}$ and
$\mathbf{y}$, which is the number of positions at which the corresponding
values are different.

\begin{theorem}\label{thm:hamming}
For any Boolean domain $D$, and any $\mathcal{F}:D^k \rightarrow D^k$, the following are equivalent:
\begin{enumerate}
\item $\mathcal{F}\in\mul{\Gamma_{\sf sub, 2}}$.
\item $\mathcal{F}\in\mul{\Gamma_{\sf sub, 2}^{\infty}}$, where $\Gamma_{\sf sub,2}^{\infty}$ denotes
the set of binary submodular cost functions taking finite {\em or infinite\/} values.
\item $\mathcal{F}$ is conservative and Hamming distance non-increasing.
\end{enumerate}
\end{theorem}
\begin{proof}
First we consider unary cost functions.
All unary cost functions on a Boolean domain are easily shown to be submodular.
Also, any conservative function $\mathcal{F}:D^k\rightarrow D^k$ is clearly a multimorphism
of any unary cost function, since it merely permutes its arguments.

For any $d \in D$ and $c \in \mathbb{R}$, define the unary cost function $\mu_c^d$ as follows:
\begin{equation*}
\mu_c^d(x)\ =\ \begin{cases}
c & \mbox{if $x=d$}, \\
0 & \mbox{if $x\neq d$}.
\end{cases}
\end{equation*}

Let $\mathcal{F}:D^k\rightarrow D^k$ be a non-conservative function.
In that case, there are $u_1,\ldots,u_k,v_1,\ldots,v_k\in D$ such that
$\mathcal{F}(u_1,\ldots,u_k)=\tuple{v_1,\ldots,v_k}$ and there is $i$ such that
$v_i$ occurs more often in $\tuple{v_1,\ldots,v_k}$ than in
$\tuple{u_1,\ldots,u_k}$. It is simple to check that
$\mathcal{F}$ is not a multimorphism of the unary cost function $\mu_1^{v_i}$.
Hence any $\mathcal{F} \in \mul{\Gamma_{\sf sub, 2}}$ must be conservative.

By the same argument,
any $\mathcal{F} \in \mul{\Gamma_{\sf sub, 2}^{\infty}}$ must be conservative.

For any $c \in \RRo$, define the binary cost functions $\lambda_c$ and $\chi_c$ as follows:
\begin{equation*}
\lambda_c(x,y)\ =\ \begin{cases}
c & \mbox{if $x=0$ and $y=1$}, \\
0 & \mbox{otherwise}.
\end{cases}
\hspace{2cm}
\chi_{c}(x,y)\ =\ \begin{cases}
c  &  \mbox{if $x \neq y$,} \\
0  &  \mbox{otherwise.}
\end{cases}
\end{equation*}
Note that $\chi_c(x,y) = \lambda_c(x,y) + \lambda_c(y,x)$.

By a simple case analysis, it is straightforward to check that any binary submodular
cost function on a Boolean domain can be expressed by binary functions of the form  $\lambda_c$, 
with $c > 0$
together with unary cost functions of the form $\mu_c^d$.

We observe that when $c < \infty$, $\lambda_c(x,y) = (\chi_c(x,y)+\mu_c^0(x)+\mu_c^1(y) - c)/2$,
so $\lambda_c$ can be expressed by functions of the form $\chi_c$
together with unary cost functions of the form $\mu_c^d$.
Hence, since expressibility preserves multimorphisms,
$\mul{\Gamma_{\sf sub, 2}} = \mul{\{\chi_c \mid c \in \mathbb{R}, c>0\}}\cap\mul{\{\mu_c^d \mid c \in \mathbb{R}, d \in D\}}$.

Now let $\mathbf{u},\mathbf{v} \in D^k$,
and consider the multimorphism inequality, as given in Definition~\ref{def:multim},
for the case where $t_i = \tuple{\mathbf{u}[i],\mathbf{v}[i]}$, for $i=1,\ldots,k$.
By Definition~\ref{def:multim}, for any $c > 0$, $\mathcal{F}$ is a multimorphism of $\chi_c$
if and only if the following holds for all choices of $\mathbf{u}$ and $\mathbf{v}$:
\[
  H(\mathbf{u},\mathbf{v})\ge H(\mathcal{F}(\mathbf{u}), \mathcal{F}(\mathbf{v})).
\]
This proves that the multimorphisms of $\Gamma_{\sf sub, 2}$ are precisely
the conservative functions which are also Hamming distance non-increasing.

Since $\Gamma_{\sf sub,2}\subseteq\Gamma_{\sf sub,2}^{\infty}$, we know that
$\mul{\Gamma_{\sf sub,2}^{\infty}}\subseteq\mul{\Gamma_{\sf sub,2}}$. Therefore,
in order to complete the proof it is enough to show that every
conservative and Hamming distance non-increasing function $\mathcal{F}$ is a
multimorphism of $\lambda_{\infty}$.

For any $\mathbf{u},\mathbf{v}\in\{0,1\}^k$, the Hamming
distance $H(\mathbf{u},\mathbf{v})$ is equal to the symmetric difference
of the sets of positions where $\mathbf{u}$ and $\mathbf{v}$
take the value $1$. Hence, for tuples $\mathbf{u}$ and $\mathbf{v}$ containing some fixed number of 1s,
the minimum Hamming distance occurs precisely when
one of these sets of positions is contained in the other.

Now consider again the multimorphism inequality, as given in Definition~\ref{def:multim},
for the case where $t_i = \tuple{\mathbf{u}[i],\mathbf{v}[i]}$, for $i=1,\ldots,k$.
If there is any position $i$ where $\mathbf{u}[i]=0$ and $\mathbf{v}[i]=1$,
then $\lambda_{\infty}(t_i) = \infty$,
so the multimorphism inequality is trivially satisfied.
If there is no such position, then the set of positions where $\mathbf{v}$ takes the value 1
is contained in the set of positions where $\mathbf{u}$ takes the value 1, so $H(\mathbf{u},\mathbf{v})$
takes its minimum possible value over all reorderings of $\mathbf{u}$ and $\mathbf{v}$.
Hence if $\mathcal{F}$ is conservative,
then $H(\mathbf{u},\mathbf{v}) \leq H(\mathcal{F}(\mathbf{u}),\mathcal{F}(\mathbf{v}))$,
and if $\mathcal{F}$ is Hamming distance non-increasing, we have
$H(\mathbf{u},\mathbf{v})= H(\mathcal{F}(\mathbf{u}),\mathcal{F}(\mathbf{v}))$.
But this implies that the set of positions where $\mathcal{F}(\mathbf{v})$ takes the value 1
is contained in the set of positions where $\mathcal{F}(\mathbf{u})$ takes the value 1.
By definition of $\lambda_{\infty}$, this implies that both sides of the multimorphism inequality are zero,
so $\mathcal{F}$ is a multimorphism of $\lambda_{\infty}$.
\end{proof}

\subsection{Non-expressibility of $\Gamma_{\sf sub}$ over $\Gamma_{\sf sub,2}$}

Consider the (carefully chosen) function $\mathcal{F}_{sep}:\{0,1\}^5\rightarrow \{0,1\}^5$
defined in Figure~\ref{fig:sepdef}.
We will show in this section that this particular function can be used to characterise
all the submodular functions of arity 4 which are expressible by binary submodular
functions on a Boolean domain,
and hence show that some submodular functions are not expressible.
\begin{figure}[htb]
\[
\begin{array}{r}
\begin{array}{c|r}
&
0\ 0\ 0\ 0\ 0\ 0\ 0\ 0\ 0\ 0\ 0\ 0\ 0\ 0\ 0\ 0\ 1\ 1\ 1\ 1\ 1\ 1\ 1\ 1\ 1\ 1\ 1\
1\ 1\ 1\ 1\ 1\\
\phantom{\mathcal{F}_{sep}(\mathbf{x})} &
0\ 0\ 0\ 0\ 0\ 0\ 0\ 0\ 1\ 1\ 1\ 1\ 1\ 1\ 1\ 1\ 0\ 0\ 0\ 0\ 0\ 0\ 0\ 0\ 1\ 1\ 1\
1\ 1\ 1\ 1\ 1\\
\mathbf{x} &
0\ 0\ 0\ 0\ 1\ 1\ 1\ 1\ 0\ 0\ 0\ 0\ 1\ 1\ 1\ 1\ 0\ 0\ 0\ 0\ 1\ 1\ 1\ 1\ 0\ 0\ 0\
0\ 1\ 1\ 1\ 1 \\
&
0\ 0\ 1\ 1\ 0\ 0\ 1\ 1\ 0\ 0\ 1\ 1\ 0\ 0\ 1\ 1\ 0\ 0\ 1\ 1\ 0\ 0\ 1\ 1\ 0\ 0\ 1\
1\ 0\ 0\ 1\ 1 \\
&
0\ 1\ 0\ 1\ 0\ 1\ 0\ 1\ 0\ 1\ 0\ 1\ 0\ 1\ 0\ 1\ 0\ 1\ 0\ 1\ 0\ 1\ 0\ 1\ 0\ 1\ 0\
1\ 0\ 1\ 0\ 1 \\
\end{array}
\\
\hline
\begin{array}{c|r}
&
0\ 0\ 0\ 0\ 0\ 0\ 0\ 0\ 0\ 0\ 0\ 0\ 0\ 0\ 0\ 0\ 0\ 0\ 0\ 1\ 0\ 0\ 0\ 1\ 0\ 0\ 0\
1\ 0\ 0\ 0\ 1 \\
&
0\ 0\ 0\ 0\ 0\ 0\ 0\ 0\ 0\ 0\ 0\ 0\ 0\ 1\ 0\ 1\ 0\ 0\ 0\ 0\ 0\ 0\ 0\ 0\ 0\ 0\ 0\
0\ 0\ 1\ 1\ 1 \\
\mathcal{F}_{sep}(\mathbf{x}) &
0\ 0\ 0\ 0\ 0\ 0\ 1\ 1\ 0\ 0\ 0\ 1\ 0\ 0\ 1\ 1\ 0\ 0\ 0\ 0\ 0\ 1\ 1\ 1\ 1\ 1\ 1\
1\ 1\ 1\ 1\ 1 \\
&
0\ 0\ 0\ 1\ 0\ 1\ 0\ 1\ 1\ 1\ 1\ 1\ 1\ 1\ 1\ 1\ 1\ 1\ 1\ 1\ 1\ 1\ 1\ 1\ 1\ 1\ 1\
1\ 1\ 1\ 1\ 1 \\
&
0\ 1\ 1\ 1\ 1\ 1\ 1\ 1\ 0\ 1\ 1\ 1\ 1\ 1\ 1\ 1\ 0\ 1\ 1\ 1\ 1\ 1\ 1\ 1\ 0\ 1\ 1\
1\ 1\ 1\ 1\ 1 \\
\end{array}
\end{array}
\]
\caption{Definition of $\mathcal{F}_{sep}$.} \label{fig:sepdef}
\end{figure}
\begin{proposition} \label{prop:sepmul0}
$\mathcal{F}_{sep}$ is conservative and Hamming distance non-increasing.
\end{proposition}
\begin{proof}
Straightforward exhaustive verification.
\end{proof}

\begin{theorem}\label{thm:sepmul}
For any function $f \in \Gamma_{\sf sub,4}$ the following are equivalent:
\begin{enumerate}
\item $f \in \express{\Gamma_{\sf sub,2}}$;
\item $\mathcal{F}_{sep} \in \mul{\{f\}}$;
\item $f \in \cone{\Gamma_{\sf fans,4}}$.
\end{enumerate}
\end{theorem}
\begin{proof}
Proposition~\ref{prop:sepmul0} and Theorem~\ref{thm:hamming} imply that
$\mathcal{F}_{sep}$ is a multimorphism of any binary submodular function on a Boolean domain.
Hence having $\mathcal{F}_{sep}$ as a multimorphism is
a necessary condition for any submodular cost function on a Boolean domain
to be expressible by binary submodular cost functions.

We will now complete the proof by showing that for 4-ary submodular cost
functions on a Boolean domain having $\mathcal{F}_{sep}$ as a multimorphism is
also {\em sufficient\/} to ensure expressibility by binary cost functions.

We consider the complete set of inequalities on the values of a 4-ary cost
function resulting from having the multimorphism $\mathcal{F}_{sep}$, as
specified in Definition~\ref{def:multim}. Out of $16^5$ such inequalities,
there are 4635 which are distinct.
After removing from these all those which are equal to the sum of two others, we obtain a
system of just 30 inequalities which must be satisfied by any 4-ary submodular
cost function which has the multimorphism $\mathcal{F}_{sep}$. Using the double
description method\footnote{As implemented, for example, by the program {\sc
Skeleton} available from {\tt
http://www.uic.nnov.ru/\~{}zny/skeleton/}}~\cite{Motzkin53:ddm} we obtain from
these 30 inequalities an equivalent set of 31 extreme rays which generate the
same polyhedral cone of cost functions. These extreme rays all correspond to fans or sums of fans,
and hence are expressible over $\Gamma_{\sf sub,2}$, by Theorem~\ref{thm:upper}.
It follows that any cost function in this cone of functions is also expressible
over $\Gamma_{\sf sub,2}$.
\end{proof}

Next we show that there are indeed 4-ary submodular cost functions which do not
have $\mathcal{F}_{sep}$ as a multimorphism and therefore are not expressible by
binary submodular cost functions.

\begin{Definition}
For any Boolean tuple $t$ of arity 4 containing exactly 2 ones and two zeros, we define the
4-ary cost function $\theta_t$ as follows:
\begin{equation*}
\theta_{t}(x_1,x_2,x_3,x_4)\ =\ \begin{cases}
-1  &  \mbox{if $(x_1,x_2,x_3,x_4) = (1,1,1,1)$ or $(0,0,0,0)$}, \\
\phantom{-}1   &  \mbox{if $(x_1,x_2,x_3,x_4) = t$}, \\
\phantom{-}0   &  \mbox{otherwise.}
\end{cases}
\end{equation*}
\end{Definition}

Cost functions of the form $\theta_t$ were introduced in~\cite{Promislow05:supermodular},
where they are called {\em quasi-indecomposable\/} functions. We denote by
$\Gamma_{\sf qin}$ the set of all (six) quasi-indecomposable cost functions of
arity 4.
It is straightforward to check that they are submodular, but the next result shows that they
are {\em not\/} expressible by binary submodular functions.

\begin{proposition} \label{prop:sepmul}
For all $\theta\in\Gamma_{\sf qin}$, $\mathcal{F}_{sep}\not\in\mul{\{\theta\}}$.
\end{proposition}

\begin{proof}
The table in Figure~\ref{fig:sepmultable} shows that $\mathcal{F}_{sep}\not\in\mul{\{\theta_{(1,1,0,0)}\}}$.
Permuting the columns appropriately establishes the result for all other $\theta\in\Gamma_{\sf qin}$.
\begin{figure}[htb]
  \[
 \begin{array}{c}
  \begin{array}{c}
  ~\\
  ~\\
  ~\\
  ~\\
  ~\\
  \end{array}
  \\
  \begin{array}{c}
  ~\\
  ~\\
  \mathcal{F}_{sep}\\
  ~\\
  ~\\
  \end{array}
  \end{array}
  \begin{array}{c}
  \begin{array}{cccc}
    1 & 0 & 1 & 0 \\
    1 & 0 & 0 & 1 \\
    0 & 1 & 0 & 1 \\
    0 & 1 & 1 & 0 \\
    0 & 0 & 1 & 1 \\
  \end{array}
  \\
  \hline
  \begin{array}{cccc}
    0 & 0 & 1 & 0 \\
    0 & 0 & 0 & 1 \\
    1 & 1 & 0 & 0 \\
    1 & 0 & 1 & 1 \\
    0 & 1 & 1 & 1 \\
  \end{array}
  \\
  \end{array}
  \begin{array}{c}
  \stackrel{\theta_{(1,1,0,0)}}{\longrightarrow}
  \left.
  \begin{array}{c}
    0\\
    0\\
    0\\
    0\\
    0\\
  \end{array}
  \right\}\mbox{\normalsize{$\sum$ = 0}}
  \\
  \stackrel{\theta_{(1,1,0,0)}}{\longrightarrow}
  \left.
  \begin{array}{c}
    0\\
    0\\
    1\\
    0\\
    0\\
  \end{array}
  \right\}\mbox{\normalsize{$\sum$ = 1}}
  \\
  \end{array}
  \]
  \caption{$\mathcal{F}_{sep}\not\in\mul{\{\theta_{(1,1,0,0)}\}}.$}
  \label{fig:sepmultable}
\end{figure}
\end{proof}

\begin{corollary}\label{cor:notexpr}
For all $\theta\in\Gamma_{\sf qin}$, $\theta \not\in \express{\Gamma_{\sf sub,2}}$.
\end{corollary}

\begin{proof}
By Theorem~\ref{thm:sepmul} and Proposition~\ref{prop:sepmul}.
\end{proof}

Are there any other 4-ary submodular cost functions which are not expressible
over $\Gamma_{\sf sub,2}$? Promislow and Young characterised the extreme rays of
the cone of all 4-ary submodular\footnote{In fact,
\cite{Promislow05:supermodular} studied {\em supermodular\/} cost functions, but
as $f$ is supermodular if and only if $-f$ is submodular, the results translate
easily.} cost functions and established that $\Gamma_{\sf
sub,4}=\cone{\Gamma_{\sf fans,4}\cup\Gamma_{\sf qin}}$ -- see
Theorem~5.2~of~\cite{Promislow05:supermodular}. Hence the results in this
section characterise the expressibility of all 4-ary submodular functions.

Promislow and Young conjectured that for $k \neq 4$, all extreme rays of $\Gamma_{\sf sub,k}$ are
fans~\cite{Promislow05:supermodular}. However, if this conjecture were true it would imply
that all submodular functions of arity 5 and above were expressible by binary submodular functions,
by Theorem~\ref{thm:upper}. This is clearly not the case, because inexpressible cost functions
such as those identified in Corollary~\ref{cor:notexpr} can be extended to larger arities
(e.g., by adding dummy arguments) and remain inexpressible. Hence our results refute this conjecture.
However, we suggest that this conjecture can be refined to a similar statement
concerning just those submodular functions which are expressible by binary submodular functions,
as follows:
\begin{conjecture}
For all $k$, $\Gamma_{\sf sub,k} \cap \express{\Gamma_{\sf sub,2}}  = \cone{\Gamma_{\sf fans,k}}$.
\end{conjecture}
This conjecture was previously known to be true for $k \leq 3$~\cite{Promislow05:supermodular};
Theorem~\ref{thm:sepmul} confirms that it holds for $k=4$. 

Next we show that we can test efficiently whether a submodular polynomial of
degree 4 is expressible by quadratic submodular polynomials.

\begin{Definition} \label{def:condC}

Let $p(x_1,x_2,x_3,x_4) $ be the polynomial representation of a 4-ary submodular
cost function $f$. We denote by $a_I$ the coefficient of the term $\prod_{i\in I} x_i$.
We say that $f$ satisfies condition Sep if for each $\{i,j\},\{k,l\}\subset \{1,2,3,4\}$,
with $i,j,k,l$ distinct, we have
$a_{\{i,j\}}+a_{\{k,l\}}+a_{\{i,j,k\}}+a_{\{i,j,l\}} \leq 0$.

\end{Definition}

\begin{theorem}
For any $f \in \Gamma_{\sf sub,4}$, the following are equivalent:
\begin{enumerate}
\item
$ f\in \express{\Gamma_{\sf sub,2}}$
\item
$f$ satisfies condition Sep.
\end{enumerate}
\end{theorem}
\begin{proof}
As in the proof of Theorem~\ref{thm:sepmul}, we can construct a set of 30 inequalities corresponding
to the multimorphism $\mathcal{F}_{sep}$. Each of these inequalities on the values of a
cost function can be translated into inequalities on the coefficients of the corresponding polynomial
representation. 24 of them impose the condition of submodularity,
and the remaining 6 inequalities impose condition Sep.
Hence a submodular cost function of arity 4 has the multimorphism $\mathcal{F}_{sep}$
if and only if its polynomial representation satisfies condition Sep.
The result then follows from Theorem~\ref{thm:sepmul}.
\end{proof}

\begin{corollary}
Given a submodular polynomial $p$ of degree 4, condition Sep can be used to test in
polynomial time whether $p$ is expressible by quadratic submodular polynomials.
\end{corollary}

In contrast to this result, it is known that the recognition problem for submodular polynomials of degree 4 is
co-NP-complete~\cite{Gallo88:supermodular}. Given an arbitrary polynomial of
degree 4, condition Sep recognises expressible polynomials {\em under the assumption
that the polynomial is submodular}. One might hope that submodular polynomials which are
expressible by quadratic submodular polynomials would be recognisable in polynomial
time. Unfortunately, this is not the case. In fact, as all polynomials of degree
4 used in the reduction given in~\cite{Gallo88:supermodular} satisfy condition Sep, the
original reduction from~\cite{Gallo88:supermodular} proves the following:

\begin{proposition}

Given an arbitrary polynomial $p$ of degree 4, it is co-NP-complete to test whether $p$ is a
submodular polynomial which is expressible by quadratic submodular polynomials.

\end{proposition}

\subsection{Applications}
\label{sec:applications}

As mentioned above, testing submodularity is co-NP-complete even for polynomials
of degree~4~\cite{Gallo88:supermodular}.
However, for many of the optimisation problems arising in practice, testing for submodularity
is not an issue because the function to be minimised is presented as a sum of
functions of bounded arity. In such cases, each of the bounded-arity sub-functions can be tested for
submodularity in constant time.
For example, in constraint satisfaction problems and computer vision, each
instance is specified as a sum of bounded-arity functions and these can be independently
tested for submodularity.
The recognition of submodularity only becomes co-NP-complete when a function is presented
without a fixed decomposition into sub-functions of this kind.

\paragraph{Artificial Intelligence}

First we formally define valued constraint satisfaction
problems~\cite{Schiex95:valued,Bistarelli99:semiring+VCSP,Rossi06:handbook}.

\begin{Definition} \label{def:vcsp}

An instance $\mathcal{P}$ of $\vcspo$ is a triple $\tuple{V,D,\mathcal{C}}$,
where $V$ is a finite set of {\em variables}, which are to be assigned values
from the set $D$, and $\mathcal{C}$ is a set of {\em valued constraints}. Each
$c \in \mathcal{C}$ is a pair $c=\tuple{\sigma,\phi}$, where $\sigma$ is a tuple
of variables of length $|\sigma|$, called the {\em scope\/} of $c$, and
$\phi:D^{|\sigma|}\rightarrow\RRo$ is a cost function. An {\em assignment\/} for
the instance $\mathcal{P}$ is a mapping $s$ from $V$ to $D$. The {\em cost\/} of
an assignment $s$ is defined as follows: \[
Cost_{\mathcal{P}}(s)=\sum_{\tuple{\tuple{v_1,v_2,\ldots,v_m},\phi} \in\C}
\phi(\tuple{s(v_1),s(v_2),\ldots,s(v_m)}). \] A {\em solution\/} to
$\mathcal{P}$ is an assignment with minimum cost.

\end{Definition}

Now we show how our results can be applied in this framework.

\begin{corollary}[of Theorem~\ref{thm:upper}] \label{cor:upper}

$\vcsp{\Gamma_{\sf fans}}$ is solvable in $O((n+k)^3)$ time, where where $n$ is the
number of variables and $k$ is the number of higher-order (ternary and above)
constraints.

\end{corollary}

Moreover, as shown above,$\vcsp{\Gamma_{\sf fans,4}}$ is the {\em maximal\/} class in
$\vcsp{\Gamma_{\sf sub,4}}$ which can be solved by reduction to  {\sc Min-Cut}
in this way.

Cohen et al.~\cite{Cohen06:expressive} showed that if a cost function
$\phi$ of arity $k$ is expressible by some set of cost functions over $\Gamma$, then $\phi$ is expressible by $\Gamma$
using at most $2^{2^k}$ extra variables. Our results show that
only $O(k)$ extra variables are needed to express any cost function
from $\Gamma_{\sf fans,k}$ by $\Gamma_{\sf sub,2}$.
Therefore, an instance of $\vcsp{\Gamma_{\sf fans}}$ needs only linearly many
(in the number of higher-order constraints) extra variables, where the
linear factor is proportional to the maximum arity of the constraints.
In particular, an instance of $\vcsp{\Gamma_{\sf sub,4}}$ is either reducible to
{\sc Min-Cut} with only linearly many extra variables,\footnote{Optimal (in the
number of extra variables) gadgets for cost functions from $\Gamma_{\sf fans,4}$
were shown in~\cite{zj08:sub-tr}.} or is not reducible at all.

\paragraph{Computer Vision}

In computer vision, many problems can be naturally formulated in terms of energy
minimisation where the energy function, over a set of variables $\{x_v\}_{v\in V}$,
has the following form:
\[
E(\vec{x})\ =\ c_0+\sum_{v\in V}c_v(x_v)+\sum_{\tuple{u,v}\in V\times
V}c_{uv}(x_u,x_v)+\ldots\]
Set $V$ usually corresponds to pixels, $x_v$ denotes the label of of pixel $v\in
V$ which must belong to a finite domain $D$. The constant term of the energy is
$c_0$, the unary terms $c_v(\cdot)$ encode data penalty functions, the pairwise
terms $c_{uv}(\cdot,\cdot)$ are interaction potentials, and so on. Functions of
arity 3 and above are also called higher-order cliques. This energy is often
derived in the context of {\em Markov Random Fields}~\cite{Geman84,Besag86}: a
minimum of $E$ corresponds to a {\em maximum a-posteriori} (MAP) labelling
$\vec{x}$~\cite{Lauritzen96,Wainwright}.

It is straightforward that this is equivalent to $\vcspo$. See~\cite{Werner07}
for a survey on the connection between computer vision and constraint
satisfaction problems. Therefore, for energy minimisation over Boolean variables
we get the following:

\begin{corollary}[of Theorem~\ref{thm:upper}] \label{cor:upper2}

Energy minimisation, where each term of the energy function belongs to $\Gamma_{\sf fans}$, is
solvable in $O((n+k)^3)$ time, where where $n$ is the number of variables (pixels) and
$k$ is the number of higher-order (ternary and above) terms in the energy function.
\end{corollary}

Note that any variable over a non-Boolean domain $D=\{0,1,\ldots,d-1\}$ of size
$d$ can be encoded by $d-1$ Boolean variables. One such encoding is the
following: $en(i)=0^{d-i-1}1^{i}$. We replace each variable with $d-1$ new
Boolean variables and impose a (submodular) relation on these new variables
which ensures that they only take values in the range of the encoding function
$en$. Note that $en(\max(a,b))=\max(en(a),en(b))$ and
$en(\min(a,b))=\min(en(a),en(b))$, so this encoding preserves submodularity.
Observe that any submodularity-preserving encoding of a non-Boolean variable by
Boolean variables needs at least $O(d)$ variables.
However, for practical purposes, subclasses of non-Boolean submodular functions
which can be encoded by Boolean submodular functions with fewer variables have
been studied, as well as approximation algorithms for these
problems~\cite{Ramalingam08:exact,Hokli08:graph}.

\section*{Acknowledgements}

The authors would like to thank Martin Cooper for fruitful discussions on
submodular functions and in particular for help with the proof of
Theorem~\ref{thm:upper}.
Stanislav \v{Z}ivn\'{y} would like to thank Philip Torr and his computer vision
group, and Tom\'{a}\v{s} Werner for clarifying the connection between constraint
satisfaction problems and computer vision.
Stanislav \v{Z}ivn\'y gratefully acknowledges the support of EPSRC grant
EP/F01161X/1.

\newcommand{\noopsort}[1]{}

\end{document}